\documentclass[conference]{IEEEtran}

\usepackage{amsmath, amssymb, mathtools}
\usepackage{graphicx, color, tabu}
\usepackage[latin1]{inputenc}
\usepackage{subcaption, subfig, tikz}
\usepackage{flushend}

\usetikzlibrary{patterns}

\newtheorem{theorem}{Theorem}

\newtheorem{proposition}{Proposition}
\newtheorem{proof}{Proof}

\def\sinr {\mbox{\scriptsize\sf SINR}}
\newcommand{\PP}{\mathbb{P}}
\newcommand{\E}{\mathbb{E}}
\newcommand{\La}{\mathcal{L}}
\newcommand{\dd}{{\rm d}}
\newcommand{\R}{\mathbb{R}}
\newcommand{\G}{\Gamma}
\newcommand{\g}{\gamma}
\newcommand{\la}{\lambda}
\newcommand{\Mod}[1]{\ (\text{mod}\ #1)}

\IEEEoverridecommandlockouts
\allowdisplaybreaks

\begin{document}

\title{Passive Loop Interference Suppression in Large-Scale Full-Duplex Cellular Networks}

\author{
\IEEEauthorblockN{Constantinos Psomas and Ioannis Krikidis}
\IEEEauthorblockA{Department of Electrical and Computer Engineering, University of Cyprus, Cyprus}
\IEEEauthorblockA{e-mail: \{psomas, krikidis\}@ucy.ac.cy}
\thanks{This work was supported by the Research Promotion Foundation, Cyprus under the project KOYLTOYRA/BP-NE/0613/04 ``Full-Duplex Radio: Modeling, Analysis and Design (FD-RD)''.}\vspace{-3mm}}

\maketitle

\begin{abstract}
Loop interference (LI) in wireless communications, is a notion resulting from the full-duplex (FD) operation. In a large-scale network, FD also increases the multiuser interference due to the large number of active wireless links that exist. Hence, in order to realize the FD potentials, this interference needs to be restricted. This paper presents a stochastic geometry model of FD cellular networks where the users and base stations employ directional antennas. Based on previous experimental results, we model the passive suppression of the LI at each FD terminal as a function of the angle between the two antennas and show the significant gains that can be achieved by this method. Together with the reduction of multiuser interference resulting from antenna directionality, our model demonstrates that FD can potentially be implemented in large-scale directional networks.
\end{abstract}


\section{Introduction}
In large-scale cellular networks, the existence of interference at a receiver is an inevitable outcome of the concurrent operation of multiple transmitters at the same carrier frequency and time slot. It is a fundamental notion of wireless communications and thus methods to reduce it \cite{RH} and even exploit it \cite{JK} are of great interest. Conventionally, multiuser interference is restricted with the use of orthogonal channels which prevent intra-cell interference, i.e. a user (base station (BS)) experiences interference only from out-of-cell BSs (users). However, even though orthogonality assists in the reduction of multiuser interference, it limits the available spectrum. Given the increased use of wireless devices, especially mobile phones, the orthogonality schemes will be quite restrictive in the near future. Towards this direction, full-duplex (FD) is considered as a possible technology for the next generation of cellular networks.

FD radio refers to the simultaneous operation of both transmission and reception using non-orthogonal channels and hence its implementation could potentially double the spectral efficiency. Nevertheless, the use of non-orthogonal channels has the critical disadvantage of increasing the interference in a cellular network, which significantly degrades its performance \cite{CP}. Firstly, the existence of more active wireless links results in the escalation of both intra- and out-of-cell multiuser interference. Secondly, the non-orthogonal operation at a transceiver creates a loop interference (LI) between the input and output antennas. This aggregate interference at a receiver is the reason why FD has been previously regarded as an unrealistic approach in wireless communications. Indeed, the primary concern towards making FD feasible was how to mitigate the LI which has an major negative impact on the receiver's performance. Recently, many methods have been developed which successfully mitigate the LI; these methods can be active (channel-aware), e.g., \cite{TR}, passive (channel-unaware), e.g., \cite{AS1}, or a combination of the two, e.g., \cite{AS2}. With regards to the multiuser interference, a well-known approach for reducing it is the employment of directional antennas. By focusing the signal towards the receiver's direction, the antenna can increase the received power and at the same time decrease the interference it generates towards other directions. The significance of directional antennas in large-scale networks has been shown before. In \cite{AH}, the authors studied an ad-hoc network's performance under some spatial diversity methods and showed the achieved gains. Moreover, in \cite{Wang}, the impact on the performance of a downlink user in a heterogeneous cellular network with directional antennas was demonstrated. The employment of directional antennas in an FD context, provides the prospect of passively suppressing the LI with antenna separation techniques \cite{AS1}, \cite{AS2}. In \cite{CP}, FD cellular networks with omnidirectional antennas are investigated where the terminals only make use of active cancellation mechanisms.

In this paper, we consider FD cellular networks where the terminals employ directional antennas and therefore, in addition to the active cancellation, they can passively suppress the LI and also reduce the multiuser interference. The main contribution of this work is the modeling of the passive suppression as a function of the angle between the transmit and receive antennas. By deriving analytical expressions of the outage probability and the average sum rate of the network, we show the significant gains that can be achieved. The rest of the paper is organized as follows. Section \ref{sec:model} presents the network model together with the channel, interference and directional antenna model. Section \ref{sec:analysis} provides the main results of the paper and in Section \ref{sec:validation} the numerical results are presented which validate our analysis. Finally, the conclusion of the paper is given in Section \ref{sec:conclusion}.

\underline{Notation}: $\R^d$ denotes the $d$-dimensional Euclidean space, $b(x,r)$ denotes a two dimensional disk of radius $r$ centered at $x$, $\|x\|$ denotes the Euclidean norm of $x \in \R^d$, $N(A)$ represents the number of points in the area $A$, $\PP(X)$ denotes the probability of the event $X$ and $\E(X)$ represents the expected value of $X$.

\section{System Model}\label{sec:model}
FD networks can be categorized into two-node and three-node architectures \cite{JSAC}. The former, referred also as bidirectional, describes the case where both nodes, i.e., the user and the base station (BS), have FD-capabilities. The latter, also known as relay or BS architecture, describes the case where only the BS (or in other scenarios the relay) has FD-capabilities. In what follows, we consider both architectures in the case where each node employs a number of directional antennas.

\subsection{Network Model}
The network is studied from a large-scale point of view using stochastic geometry \cite{HAEN}. The locations of the BSs follow a homogeneous Poisson point process (PPP) $\Phi = \{ x_i: i = 1,2,\dots\}$ of density $\lambda$ in the Euclidean plane $\R^2$, where $x_i \in \R^2$ denotes the location of the $i^{\rm th}$ BS. Similarly, let $\Psi = \{ y_i: i = 1,2,\dots \}$ be a homogeneous PPP of the same density $\lambda$ but independent of $\Phi$ to represent the locations of the users. Assume that all BSs transmit with the same power $P_b$ and all users with the same power $P_u$. A user selects to connect to the nearest BS in the plane, that is, BS $i$ serves user $j$ if and only if $\|x_i-y_j\| < \|x_i-y_k\|$ where $y_k \in \Psi$ and $k \neq j$. Assuming the user is located at the origin $o$ and at a distance $r$ to the nearest BS, the probability density function (pdf) of $r$ is $f_r(r) = 2\pi \lambda re^{-\lambda \pi r^2}, ~r \geq 0$ \cite{HAEN}. Note that this distribution is also valid for the nearest distance between two users and between two BSs.

\subsection{Channel Model}
All channels in the network are assumed to be subject to both small-scale fading and large-scale path loss. Specifically, the fading between two nodes is Rayleigh distributed and so the power of the channel fading is an exponential random variable with mean $1/\mu$. The channel fadings are considered to be independent between them. The standard path loss model $\ell(x,y) = \|x-y\|^{-\alpha}$ is used which assumes that the received power decays with the distance between the transmitter $x$ and the receiver $y$, where $\alpha > 2$ denotes the path loss exponent. Throughout this paper, we will denote the path loss exponent for the channels between a BS and a user by $\alpha_1$. For the sake of simplicity, we will denote by $\alpha_2$ the path loss exponent for the channels between BSs and between users. Lastly, we assume all wireless links exhibit additive white Gaussian noise (AWGN) with zero mean and variance $\sigma_n^2$.

\subsection{Sectorized Directional Antennas}
Define as $M_b$ and $M_u$ the number of directional transmit/receive antennas employed at a BS and a user respectively. The main and side lobes of each antenna are approximated by a circular sector as in \cite{AH}. Therefore, the beamwidth of the main lobe is $2\pi/M_i$, $i \in \{b, u\}$. It is assumed that an active link between a user and a BS lies in the boresight direction of the antennas of both nodes, i.e., maximum power gain can be achieved. Note that $M_b = M_u = 1$ refers to the omni-directional case \cite{CP}. As in \cite{AH}, we assume that the antenna gain of the main lobe is $G_i = \left(\frac{M_i}{1+\g_i(M_i-1)}\right)$ where $\g_i$, $i \in \{b, u\}$ is the ratio of the side lobe level to the main lobe level. Therefore, the antenna gain of the side lobe is $H_i = \g_i G_i$, $i \in \{b, u\}$.

\subsection{Multiuser Interference}\label{subsec:inter}
The total multiuser interference at a node is the aggregate sum of the interfering received signals from the BSs of $\Phi$ and the uplink users of $\Psi$. In the two-node architecture, multiuser interference at any node results from both out-of-cell users and BSs. In the three-node architecture, the BS experiences multiuser interference from out-of-cell BSs and users, whereas the downlink user experiences additional intra-cell interference from the uplink user. When $M_b > 1$ or $M_u > 1$ the transmitters can interfere with a receiver in four different ways \cite{AH}:
\begin{itemize}
\item[1.] Transmitting towards a receiver in the main sector,
\item[2.] Transmitting away from a receiver in the main sector,
\item[3.] Transmitting towards a receiver outside the main sector,
\item[4.] Transmitting away from a receiver outside the main sector,
\end{itemize}
where the main sector is the area covered by the main lobe of the receiver.

Consider the interference received at a node $x^i \in \Phi \cup \Psi$ from all other network nodes $y^j \in \Phi \cup \Psi$, $i,j \in \{b,u\}$, $x^i \neq y^j$. To evaluate the interference, each case $k \in \{1,2,3,4\}$ needs to be considered separately. This results in each of the PPPs $\Phi$ and $\Psi$ being split into four thinning processes $\Phi_k$ and $\Psi_k$ with densities $\la_{i,j,k}$. Additionally, the power gain $\G_{i,j,k}$ of the link between $x^i$ and $y^j$ changes according to $k$. Table \ref{tbl:thin} provides the density and power gain for each case. Note that $\sum_{k=1}^4 \la_{i,j,k} = \la$ and when $M_b = M_u = 1$ the links have no gain, i.e., $\G_{i,j,k}=1 ~\forall~ i,j,k$.

\begin{table}[t]\centering\tabulinesep=1mm
  \begin{tabu}{| c | c | c | c | c |}\hline
    $k$ & 1 & 2 & 3 & 4\\ \hline
    $\la_{i,j,k}$ & $\frac{\la}{M_i M_j}$ & $\frac{\la(M_j-1)}{M_i M_j}$ & $\frac{\la(M_i-1)}{M_i M_j}$ & $\frac{\la(M_i-1)(M_j-1)}{M_i M_j}$\\ \hline
    $\G_{i,j,k}$ & $G_iG_j$ & $G_iH_j$ & $G_jH_i$ & $H_iH_j$\\ \hline
  \end{tabu}
\caption{Densities $\la_{i,j,k}$ and power gains $\G_{i,j,k}$ for each thinning process $k \in \{1,2,3,4\}$, $i,j \in \{b, u\}$.}\label{tbl:thin}
\end{table}

\section{Performance Analysis}\label{sec:analysis}
In this section, we analytically derive the outage probability and the sum rate of an FD cellular network implementing the three-node architecture. The respective expressions for the two-node architecture are omitted since they can be derived in a similar way. The performance analysis is derived using similar procedures as in \cite{CP}, \cite{AH} and \cite{JA}. Without loss of generality and following Slivnyak's Theorem \cite{HAEN}, we execute the analysis for a typical node located at the origin but the results hold for all nodes in the network. Denote by $u_o$, $u'_o$ and $b_o$ the typical downlink user, uplink user and BS respectively.

\begin{figure}
\begin{subfigure}{0.45\linewidth}\centering
\begin{tikzpicture}[scale=0.4]
  \foreach \x/\text in {0.1,0.2,...,1} \draw [dashed,thin,red] circle(\x);
  \foreach \x/\text in {0.01,0.02,...,3.5} {
    \draw [-,thin,blue] (-22.5:\x) arc(-22.5:22.5:\x);
  }
  \foreach \x/\text in {0.1,0.3,...,3.5} {
    \draw [-,thin,green] (67.5:\x) arc(67.5:112.5:\x);
  }
  
  \foreach \ang/\lab/\dir in {
    0/{0}/right,
    1/{\frac{\pi}{4}}/above right,
    2/{\frac{\pi}{2}}/above,
    3/{\frac{3\pi}{4}}/above left,
    4/{\pi}/left,
    5/{\frac{5\pi}{4}}/below left,
    6/{\frac{3\pi}{2}}/below,
    7/{\frac{7\pi}{4}}/below right} {
    \draw (0,0) -- (\ang*45:3.7); \draw[densely dotted] (0,0) -- (22.5+\ang*45:3.5);
    \node [fill=white] at (\ang*45:3.8) [\dir] {\scriptsize $\lab$};
  }
  
  \draw [<->,ultra thin] (0:5.1) arc(0:90:5.1);
  \node [fill=white] at (45:5.1) [above right] {\scriptsize $\theta$};
\end{tikzpicture}
\caption{}\label{fig:ant_angle}
\end{subfigure}\hfill
\begin{subfigure}{0.45\linewidth}\centering
\begin{tikzpicture}[scale=0.35]

\draw[dashed, ultra thin] (0,0) circle (1);
\draw[dashed, ultra thin] (0,0) circle (2); \node [fill=white] at (-2, 0) [below] {\tiny $0.5$};
\draw[dashed, ultra thin] (0,0) circle (3);
\draw[style=double, ultra thin] (0,0) circle (4);
\node[fill=white] at (-4, 0) [below] {\tiny $1$};

\foreach \ang/\lab/\dir in {
  0/0/right,
  1/{\frac{\pi}{4}}/{above right},
  2/{\frac{\pi}{2}}/above,
  3/{\frac{3\pi}{4}}/{above left},
  4/{\pi}/left,
  5/{\frac{5\pi}{4}}/{below left},
  7/{\frac{7\pi}{4}}/{below right},
  6/{\frac{3\pi}{2}}/below} {
  \draw[densely dashed, ultra thin] (0,0) -- (\ang * 180 / 4:4.2);
  \node [fill=white] at (\ang * 180 / 4:4.2) [\dir] {\scriptsize $\lab$};
}

\draw[red!50!black, semithick] plot [domain=-180:180]
  (xy polar cs:angle=\x, radius={4*exp(-cos(abs(\x)-2*180/3)-0.5)});
\node [fill=white] at (1.5,2) {\tiny $f_l(\theta)$};
\end{tikzpicture} 
\caption{}\label{fig:pass_cancel}
\end{subfigure}
\caption{(a) Angle $\theta$ between antennas for $M_b = 8$. Dots correspond to the boundaries of each sector. Shaded area, solid lines and dashed lines depict the main transmission lobe, main reception lobe and side lobes respectively. (b) LI passive suppression efficiency with respect to the angle $\theta$. The value of $f_l(\theta)$ corresponds to the fraction of the LI which cannot be suppressed, that is $f_l(\theta) = 1$ means no suppression.}
\end{figure}
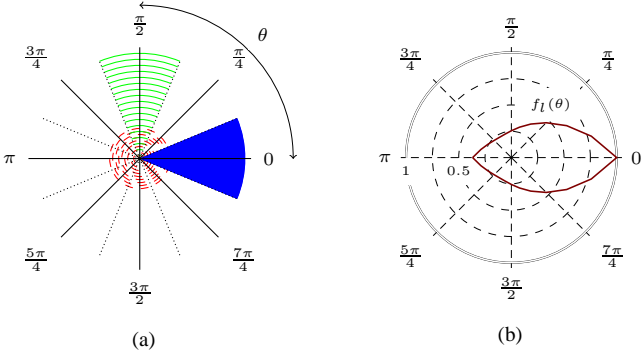

\subsection{Performance Metrics}
The outage probability describes the probability that the instantaneous achievable rate of the channel is less than a fixed target rate $R$, i.e., $\PP[\log(1+\sinr) < R]$, where SINR is the signal-to-interference-plus-noise-ratio. Assume $b_o$ is at a random distance $r$ from $u'_o$, its nearest uplink user. Let $I_u$ and $I_b$ be the aggregate interference received at $b_o$ from the uplink users (apart from $u'_0$) and the out-of-cell BSs respectively. Then $I_u$ and $I_b$ can be expressed as,
\begin{align}
I_u &= P_u \sum_{i \in \{1,2,3,4\}} \sum_{j\in\Psi_i \setminus u'_o} \G_{b,u,i} k_j D_j^{-\alpha_1},\label{eq:intu_u}\\
I_b &= P_b \sum_{i \in \{1,2,3,4\}} \sum_{j\in\Phi_i} \G_{b,b,i} g_j d_j^{-\alpha_2},\label{eq:intb_u}
\end{align}
where $D_j$ and $k_j$ denote the distance and the channel fading respectively between $b_o$ and the $j^{\rm th}$ uplink user; similarly, $g_j$ and $d_j$ denote the distance and the channel fading between $b_o$ and the $j^{\rm th}$ BS.

Imperfect active cancellation mechanisms are used at the BSs to reduce the effects of the LI. Each implementation of the cancellation mechanism can be described by a certain residual power and modeled as a fading feedback channel. We denote the channel gain at a node from the residual LI as $h_l$ and assume that it follows a complex Gaussian distribution with zero mean and variance $\sigma_l^2$ \cite{TR}. Furthermore, we assume that each BS employs antenna separation techniques to passively suppress the LI. Let $\theta \in [-\pi, \pi)$ be the angle between the two antennas (Fig. \ref{fig:ant_angle}). By adopting the results in \cite{AS1} we assume that the maximum suppression is achieved at $\theta = \frac{2\pi}{3}$. Even though the authors of \cite{AS1} also take into account the distance between the BS and the user, we mainly focus on the impact of the angle but our model can be easily extended to incorporate the distance. Observe that passive suppression in the two-node architecture is not possible due to its bidirectional property, i.e., both antennas point to the same direction. Let $f_l(\theta)$ denote the fraction of the LI that cannot be passively suppressed at an angle $\theta$, e.g., $f_l(\theta) = 1$ means zero suppression, and it is written as,
\begin{equation}
f_l(\theta) = \exp\left(-\cos\left(|\theta|-\frac{2\pi}{3}\right)-\frac{1}{2}\right),
\end{equation}
where $\theta \in [-\pi, \pi), \theta \equiv 0 \Mod{\frac{2\pi}{M_b}}$. Fig. \ref{fig:pass_cancel} depicts the level of passive suppression with respect to the angle $\theta$. Denote by $I_l$ the total channel gain from the LI at $b_o$ after active cancellation and passive suppression. Based on the above assumptions,
\begin{equation}\label{eq:li}
I_l = P_u G_b^2 h_l f_l(\theta)(B_0 + \g_b (1-B_0)),
\end{equation}
where $h_l \sim \exp(1/\sigma^2_l)$ and $B_0 \sim {\rm Bernoulli}(\frac{1}{M_b})$ is a binary random variable with
\begin{equation}\label{eq:bern}
B_0 =
\begin{cases}
1 & {\rm with ~prob.~}~ \frac{1}{M_b} \,~~~~(\theta = 0),\\
0 & {\rm with ~prob.~} \frac{M_b-1}{M_b} ~~~(\theta \neq 0).
\end{cases}
\end{equation}
since the power gain of the LI signal is $G_b^2$ for $\theta = 0$ and $G_b H_b$ otherwise. Then, the $\sinr$ at $b_o$ can be defined as,
\begin{equation}\label{eq:sinr_u}
\sinr^u = \frac{P_u \G_{b,u,1} h r^{\alpha_1}}{\sigma_n^2 + I_l + I_b + I_u},
\end{equation}
where $h$ is the channel fading between $b_o$ and $u'_o$.

Note that in \eqref{eq:li} we consider the active cancellation and passive suppression of the LI separately. However, in reality, the active cancellation mechanism attempts to mitigate the passively suppressed LI and therefore a more ``realistic" model would be to express the variance $\sigma^2_l$ as a function of $f_l(\theta)$. For the sake of simplicity, we assume that $f_l(\theta)$ is a normalization factor of $h_l$ which makes no difference in the final results. For the outage probability, we state the following theorem.

\begin{theorem}\label{thm:up}
The outage probability at the uplink in a three-node architecture is,
\begin{align}\nonumber
&\Pi^u(R, \lambda, M_b, M_u, \alpha_1, \alpha_2)\\ &= 1 - 
(2 \pi \lambda)^2 \int_0^\infty r e^{-\lambda \pi r^2-s\sigma_n^2}
\La_{I_l}(s) \La_{I_b}(s) \La_{I_u}(s) \dd r,
\end{align}
where $s = \frac{\mu Tr^{\alpha_1}}{P_u G_b G_u}$, $T=2^R-1$,
\begin{align}
\La_{I_l}(s)=&
\frac{1}{M_b}\Bigg[\frac{1}{1+\frac{P_b G_b}{P_u G_u}\mu\sigma_l^2 Tr^{\alpha_1}}\nonumber\\
&+\sum\limits_{\substack{\theta \in [-\pi, \pi) \setminus \{0\}\\\theta \equiv 0 \Mod{\frac{2\pi}{M_b}}}}\frac{1}{1+\frac{P_b H_b}{P_u G_u}\frac{\mu\sigma_l^2 Tr^{\alpha_1}}{\exp(\cos(|\theta|-\frac{2\pi}{3})+\frac{1}{2})}}\Bigg],\\\label{eq:laplace_li_3u}
\La_{I_b}(s)=&\int_0^\infty \rho e^{-\la \pi \rho^2}\prod_{\mathclap{i\in\{1,2,3,4\}}}\exp\Bigg(-2\pi\la_{b,b,i}\nonumber\\ &\times\int_\rho^\infty \left(\frac{\frac{P_b}{P_u}c_iT}{\frac{P_b}{P_u}c_i T + \frac{x^{\alpha_2}}{r^{\alpha_1}}} \right) x \dd x \Bigg)\dd \rho,
\end{align}
\begin{align}
&\La_{I_u}(s)=\prod_{\mathclap{i\in\{1,2,3,4\}}}\exp\left(-2\pi\la_{b,u,i} \int_r^\infty \left(\frac{q_i T}{q_i T + (\frac{y}{r})^{\alpha_1}} \right) y \dd y \right),
\end{align}
and $c_i = \frac{\G_{b,b,i}}{\G_{b,u,1}}$, $q_i = \frac{\G_{b,u,i}}{\G_{b,u,1}}$.
\end{theorem}\medskip

\begin{proof}
See Appendix.
\end{proof}

\begin{figure*}[t]
  \begin{minipage}{0.32\linewidth}\centering
    \includegraphics[width=\linewidth]{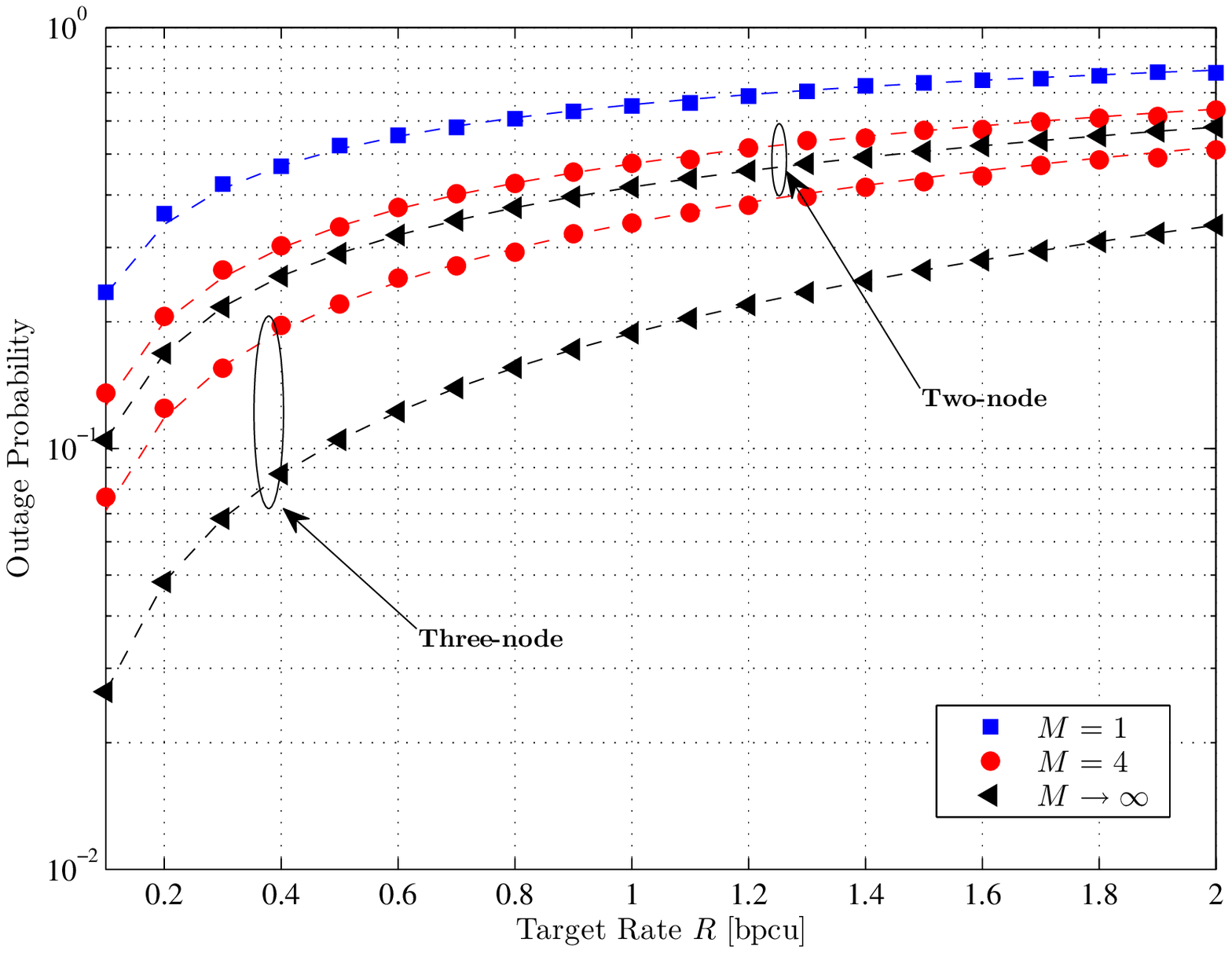}
    \captionof{figure}{Impact of directionality on the uplink outage probability; $\sigma_l^2=-30$dB.}\label{fig:outage_vs_bpcu}
  \end{minipage}\hspace{2mm}
  \begin{minipage}{0.32\linewidth}\centering
    \includegraphics[width=\linewidth]{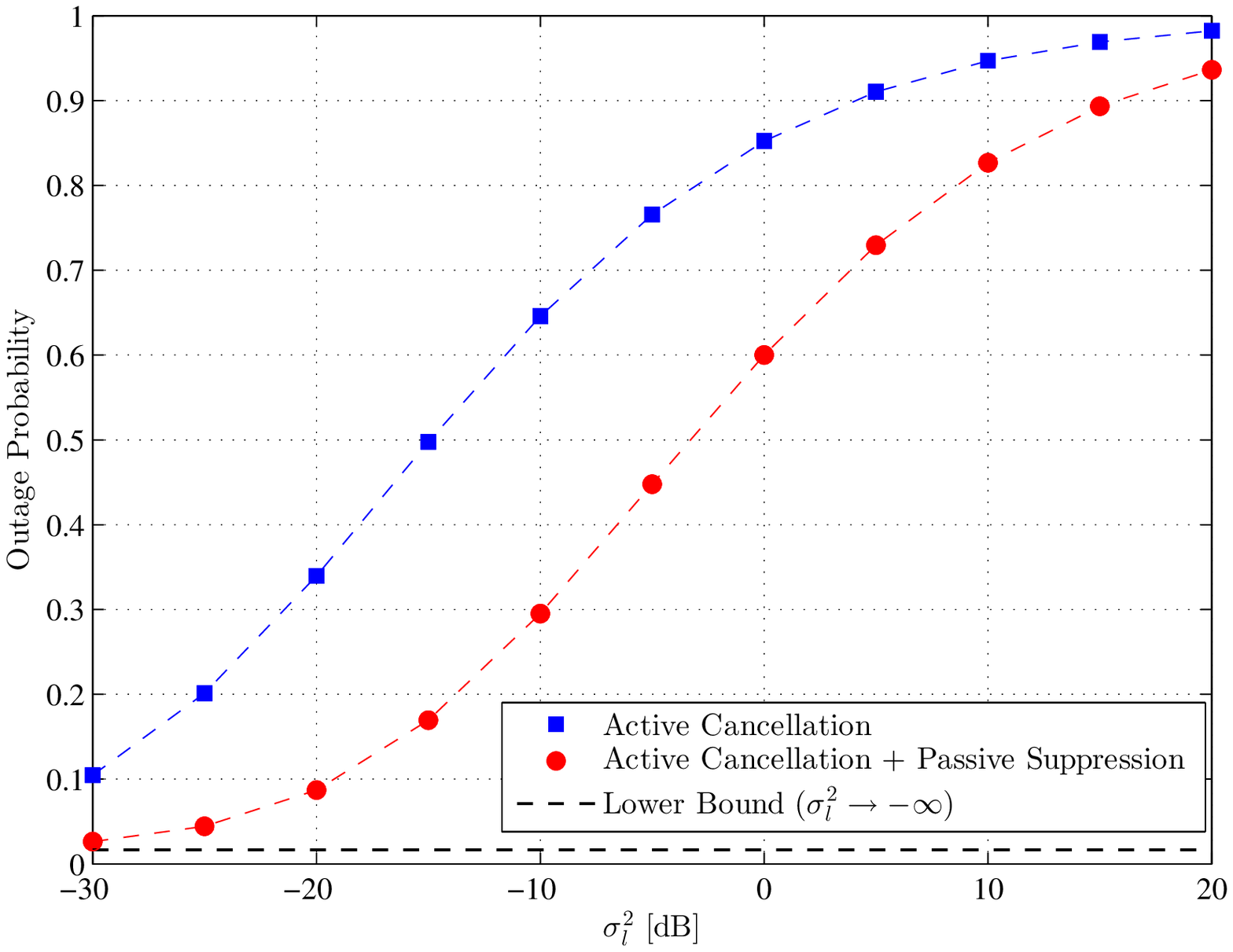}
    \captionof{figure}{Impact of passive suppression on the outage probability; $R = 0.1$ bpcu.}\label{fig:outage_vs_li}
  \end{minipage}\hspace{2mm}
  \begin{minipage}{0.32\linewidth}\centering
    \includegraphics[width=\linewidth]{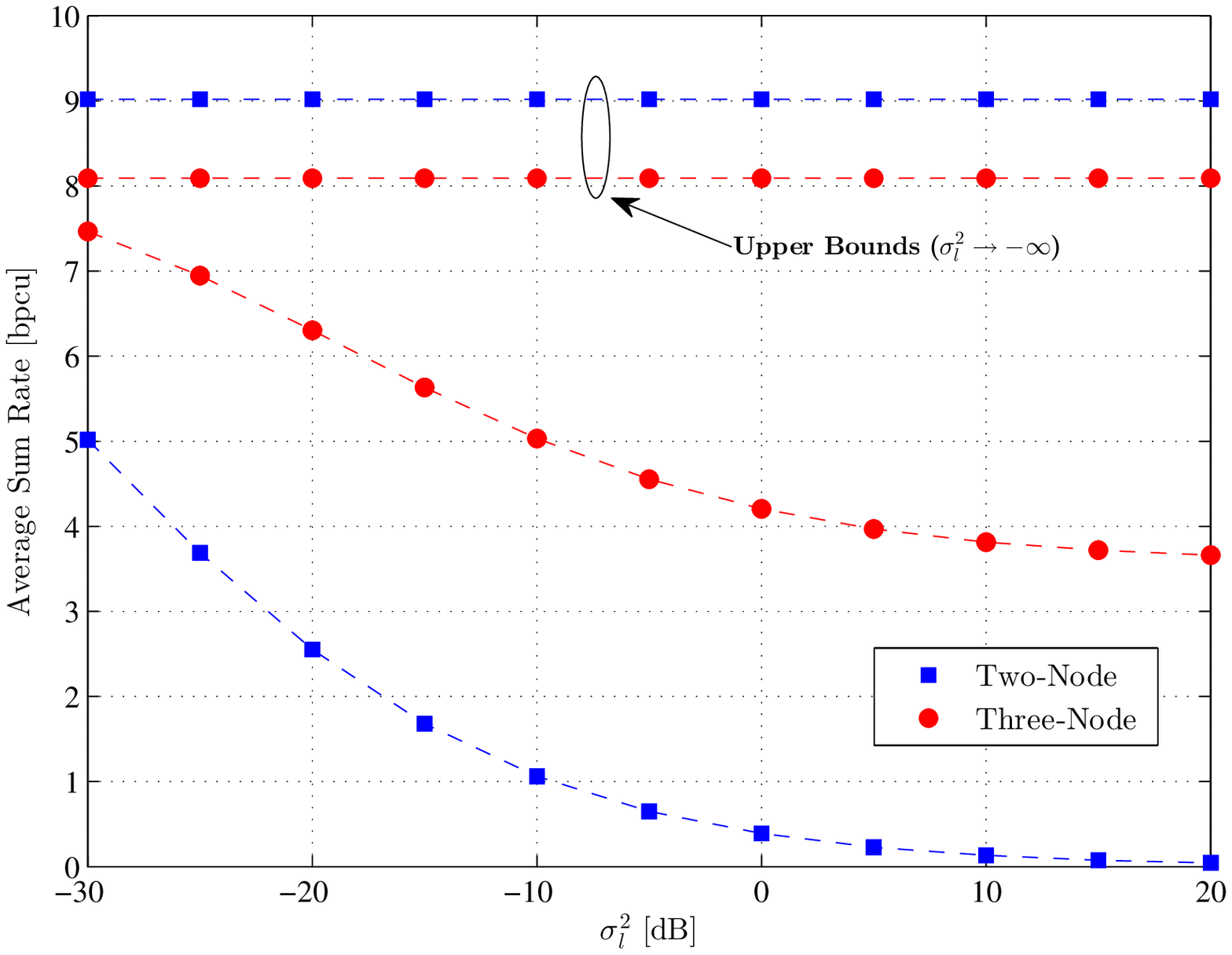}
    \captionof{figure}{Average sum rate vs $\sigma_l^2$; $R = 0.1$ bpcu.}\label{fig:rate_vs_li}
  \end{minipage}
\end{figure*}

The multiuser interference experienced by the downlink user $u_o$ is explained in Section \ref{subsec:inter}. Furthermore, $u_o$ operates in half-duplex mode and therefore does not experience any LI, i.e., $I_l = 0$. We omit the SINR equation at the downlink and the proof of the outage probability since these can be derived in a similar way as above and state only the final result.
\begin{theorem}\label{thm:three_d}
The outage probability of a downlink user is,
\begin{align}
&\Pi^d(R, \lambda, M_b, M_u, \alpha_1, \alpha_2)\nonumber\\ &= 1 -
2 \pi \lambda \int_0^\infty r e^{-\lambda \pi r^2-s\sigma_n^2}
\La_{I_b}(s) \La_{I_u}(s) \dd r,
\end{align}
where $s = \frac{\mu Tr^{\alpha_1}}{P_b G_b G_u}$, $T=2^R-1$,
\begin{align}
&\La_{I_b}(s) = \prod_{\mathclap{i\in\{1,2,3,4\}}}\exp\left(-2\pi\la_{u,b,i} \int_r^\infty \left(\frac{c_i T}{c_i T + (\frac{x}{r})^{\alpha_1}} \right) x \dd x \right),\\
&\La_{I_u}(s) = \prod_{\mathclap{i\in\{1,2,3,4\}}}\exp\left(-2\pi\la_{u,u,i} \int_0^\infty \left(\frac{\frac{P_u}{P_b} q_i T}{\frac{P_u}{P_b} q_i T + \frac{y^{\alpha_2}}{r^{\alpha_1}}} \right) y \dd y \right),
\end{align}
and $c_i = \frac{\G_{u,b,i}}{\G_{u,b,1}}$, $q_i = \frac{\G_{u,u,i}}{\G_{u,b,1}}$.
\end{theorem}\medskip

Next, we provide an expression for the average sum rate $\bar{\mathcal{R}}$ of the network. The average sum rate is the sum of the expected values of the instantaneous achievable downlink rate $\bar{\mathcal{R}}^d$ and uplink rate $\bar{\mathcal{R}}^u$, i.e., $\bar{\mathcal{R}} = \mathcal{\bar{R}}^u + \mathcal{\bar{R}}^d = \E[\log(1+\sinr^u)] + \E[\log(1+\sinr^d)]$.
\begin{proposition}
The average sum rate of a three-node FD cellular network is,
\begin{align}
\bar{\mathcal{R}} &= \int_0^\infty \left[1-\Pi^u(t, \lambda, M_b, M_u, \alpha_1, \alpha_2)\right]\dd t\nonumber\\&+ \int_0^\infty \left[1-\Pi^d(t, \lambda, M_b, M_u, \alpha_1, \alpha_2)\right]\dd t.
\end{align}
\end{proposition}

\begin{proof}
Since $\E(X) = \int_0^\infty \PP(X > t) \dd t$ for a positive random variable $X$, the result follows.
\end{proof}

\subsection{Asymptotic Case}\label{sec:asymptotic}
To simplify the above expressions, let $M = M_b = M_u$ and consider the asymptotic case when the number of employed antennas goes to infinity. Furthermore, let $\alpha_1 = \alpha_2 = 4$ and $\g = \g_b = \g_u$. Assume that the BSs and the users transmit with the same power, i.e., $P_b = P_u$, and consider high power transmissions which result in an interference-limited network, that is $\sigma_n^2 = 0$. Then, using firstly the transformations $\frac{1}{\sqrt{c_i T}}(\frac{x}{r})^2 \rightarrow u$ in $\La_{I_b}(s)$ and $\frac{1}{\sqrt{q_i T}}(\frac{y}{r})^2 \rightarrow v$ in $\La_{I_u}(s)$ and secondly the transformations $r^2 \rightarrow z$ and $\rho^2 \rightarrow w$, the outage probability as $M$ approaches infinity changes to,
\begin{equation}
\lim_{M\rightarrow \infty}\Pi^u = 1 - (\pi \lambda)^2 \int_0^\infty e^{-\lambda \pi z} \mathcal{I}_l \mathcal{I}_b \mathcal{I}_u \dd z,
\end{equation}
\begin{align}
&\mathcal{I}_l = \frac{1}{2\pi}\int_{-\pi}^{\pi}\frac{1}{1+\mu\sigma_l^2 T z^2 \exp(-\cos(|\theta|-\frac{2\pi}{3})-\frac{1}{2})}~\dd \theta,\\
&\mathcal{I}_b=\int_0^\infty \exp\left(-\pi\la (w + z \g \sqrt{T} \arctan(z \g\sqrt{T}/w))\right)\dd w,\\
&\mathcal{I}_u = \exp\left(-\pi\la z \sqrt{T} \g \arctan(\g\sqrt{T}) \right),
\end{align}
for the uplink and,
\begin{equation}\label{eq:spec_3d}
\lim_{M\rightarrow \infty}\Pi^d = 1 - \frac{1}{1+ \g \sqrt{T}\left(\arctan\left(\g\sqrt{T}\right)+\frac{\pi}{2}\right)},
\end{equation}
for the downlink.

It is clear that $\Pi^d$ is independent from the density $\la$ of the network and depends only on the target rate $R$ and the ratio $\g$. Also, when $\g \rightarrow 0$ then $\Pi^d \rightarrow 0$, which is expected since in this case there will be no multiuser interference. On the other hand, $\Pi^u$ heavily depends on the value of $\sigma_l^2$. When $\sigma_l^2 > 0$ and $\g \rightarrow 0$, $\Pi^u$ depends entirely on the function $\mathcal{I}_l$ and on the density $\la$, since in this case $\mathcal{I}_b = \frac{1}{\pi\la}$ and $\mathcal{I}_u = 1$. When $\sigma_l^2 \rightarrow -\infty$, $\Pi^u$ behaves similarly to $\Pi^d$: it becomes independent of $\la$ (even though it is not entirely clear) and its value converges to zero when $\g \rightarrow 0$. 

\section{Numerical Results}\label{sec:validation}
We validate and evaluate our proposed model with computer simulations. Unless otherwise stated, the simulations use the parameters from Section \ref{sec:asymptotic}, together with $\lambda = 10^{-2}$, $\mu = 1$ and $\g = 0.2$. In Figures \ref{fig:outage_vs_bpcu}-\ref{fig:rate_vs_li}, the dashed lines represent the analytical results. Firstly, we illustrate in Figure \ref{fig:outage_vs_bpcu} the impact on the uplink outage probability from the employment of directional antennas. Both architectures benefit greatly from directionality but the three-node achieves a better performance due to the BS's ability to passively suppress the LI. This is shown in Figure \ref{fig:outage_vs_li} where we depict the performance of a BS in terms of the outage probability, with and without passive suppression, for different values of $\sigma_l^2$. In the two extreme cases, $\sigma_l^2 \rightarrow -\infty$ and $\sigma_l^2 \rightarrow \infty$, the two methods have the same performance. In the former case, $\Pi^u$ converges to a constant floor and in the latter case $\Pi^u \rightarrow 1$. However, for moderate values, passive suppression provides significant gains, e.g., for $\sigma_l^2 = -20$ dB it achieves about $70\%$ reduction of $\Pi^u$. Finally, Figure \ref{fig:rate_vs_li} shows the average sum rate of each architecture with respect to $\sigma_l^2$. The sum rate of the three-node architecture is obviously greater than the sum rate of the two-node architecture. This is in part due to the LI passive suppression at the uplink but also due to the half-duplex mode at the downlink which is not affected by LI. When $\sigma_l^2 \rightarrow \infty$, the sum rate of the three-node converges to the rate of the downlink whereas the sum rate of the two-node converges to zero. On the other hand, when $\sigma_l^2 \rightarrow -\infty$, the two-node outperforms the three-node as expected since both nodes operate FD mode but this scenario is difficult to achieve which is also evident from the figure. 

\section{Conclusion}\label{sec:conclusion}
This paper has presented the impact of directional antennas on the performance of FD cellular networks. The ability of the three-node architecture to passively suppress the LI at the uplink has significant gains to its efficiency. Moreover, since the downlink user operates in half-duplex mode, the network can achieve high sum rates. The three-node architecture is regarded as the topology to be potentially implemented first in the case of FD employment in cellular networks. The main reason is the high energy requirements which FD will impose on future devices. The results of this paper, give insight as to how such an architecture will perform and provide another reason to support its implementation.

\appendix
Conditioned on the distance $r$ to the nearest uplink user we have,
\begin{align*}
&\Pi^u(R, \lambda, M_b, M_u, \alpha_1, \alpha_2) =
\E_r \left[ \PP[\log(1+\sinr^u) < R\ |\ r] \right]\\
&=1-2\pi \lambda\int_0^\infty \PP[\sinr^u \geq 2^R-1\ |\ r] \;re^{-\lambda \pi r^2}\dd r.
\end{align*}
The coverage probability $\PP[\sinr^u \geq T\ |\ r]$ is given by,
\begin{align*}
\PP[\sinr^u \geq T\ |\ r] &= \PP\left[h \geq \frac{Tr^{\alpha_1}}{P_b G_b G_u}(\sigma_n^2 + I_l + I_b + I_u)\ \Big|\ r\right]\\
&\stackrel{(a)}{=} e^{-s\sigma_n^2}\E_{I_l}\left[e^{-sI_l}\right]
\E_{I_b}\left[e^{-sI_b}\right]\E_{I_u}\left[e^{-sI_u}\right]\\
&= e^{-s\sigma_n^2}\La_{I_l}(s) \La_{I_b}(s) \La_{I_u}(s),
\end{align*}
where $s = \frac{\mu Tr^{\alpha_1}}{P_b G_b G_u}$, $T=2^R-1$ and $(a)$ follows from the fact that $h \sim \exp(\mu)$; $\La_{I_l}(s)$, $\La_{I_b}(s)$ and $\La_{I_u}(s)$ are the Laplace transforms of the random variables $I_l$, $I_b$ and $I_u$ respectively, evaluated at $s$. The Laplace transform of $I_l$ follows directly from the Bernoulli variable and from the moment generating function (MGF) of an exponential variable since $h_l \sim \exp(1/\sigma^2_l)$. Next we evaluate $\La_{I_b}(s)$ as follows. Assuming the distance to the closest BS from $b_o$ is $\rho$, $I_b$ needs to be evaluated conditioned on $\rho$. Since the densities of $\Phi$ and $\Psi$ are equal, we can assume that there is on average one BS per cell. Therefore, the Laplace transform of $I_b$ is given by,
\begin{equation}\label{eq:laplace_I_b}
\La_{I_b}(s) = \E_{I_b}[e^{-sI_b}\ |\ \rho] = \int_0^\infty \E_{I_b}[e^{-sI_b}]f_{\rho}(\rho)\dd \rho.
\end{equation}
Then $\E_{I_b}[e^{-sI_b}]$ is evaluated as follows,
\begin{align}
\nonumber &\E_{I_b}[e^{-sI_b}] = \prod_{i\in\{1,2,3,4\}} \E_{\Phi_i, g_j}\left[\prod_{j\in\Phi_i}e^{-s P_b \G_{b,b,i} g_j d_j^{-\alpha_2}}\right]\nonumber\\
&\stackrel{(a)}{=} \prod_{i\in\{1,2,3,4\}} \E_{\Phi_i}\left[\prod_{j\in\Phi_i}\E_g[e^{-s P_b \G_{b,b,i} g d_j^{-\alpha_2}}]\right]\nonumber\\
\nonumber&\stackrel{(b)}{=} \prod_{i\in\{1,2,3,4\}}e^{-2\pi\la_{b,b,i} \int_\rho^\infty \left(1-\E_g[\exp(-s P_b \G_{b,b,i} g y^{-\alpha_2})] \right) y \dd y}\\
&\stackrel{(c)}{=} \prod_{i\in\{1,2,3,4\}}e^{-2\pi\la_{u,u,i} \int_\rho^\infty \left(1 - \frac{\mu}{\mu + s P_u \G_{u,u,i} y^{-\alpha_2}} \right) y \dd y}\label{eq:exp_I_b},
\end{align}
where $(a)$ follows from the fact that $g_j$ are i.i.d. and also independent from the point process $\Phi$; $(b)$ follows from the probability generating functional (PGFL) of a PPP and the limits follow from the closest BS being at a distance $\rho$; $(c)$ follows from the MGF of an exponential random variable and since $g \sim \exp(\mu)$. The results follows by replacing $\E_{I_b}[e^{-sI_b}]$ with \eqref{eq:exp_I_b}, $s$ with $\frac{\mu Tr^{\alpha_1}}{P_b G_b G_u}$ in \eqref{eq:laplace_I_b} and letting $q_i=\frac{\G_{u,u,i}}{\G_{u,b,1}}$. $\La_{I_u}(s)$ can be derived in a similar way.

\end{document}